\documentclass[journal,9.7pt]{IEEEtran}
\usepackage{amsmath}
\usepackage{mathrsfs}
\usepackage{pifont}
\usepackage{bbding}
\usepackage{amsmath,epsfig}
\usepackage{stmaryrd}
\usepackage{amssymb}
\usepackage{amsfonts}
\usepackage{epic}
\usepackage{graphicx}
\usepackage{curves}
\usepackage{cite}

\usepackage{algorithm}

\usepackage{algpseudocode}

\usepackage{stfloats}
\usepackage{latexsym}
\usepackage{epstopdf}
\usepackage{epic}
\usepackage{bm}
\usepackage{multirow}
\usepackage{xcolor}

\newtheorem{theorem}{Theorem}

\newtheorem{proposition}{Proposition}

\newtheorem{remark}{Remark}

\newcommand{\A}{\mathbf A}
\newcommand{\B}{\mathbf B}

\newcommand{\ve}{\mathbf v}

\newcommand{\h}{\mathbf h}
\newcommand{\W}{\mathbf W}
\newcommand{\V}{\mathbf V}
\newcommand{\X}{\mathbf X}
\newcommand{\x}{\mathbf x}
\newcommand{\HH}{\mathbf H}
\newcommand{\0}{\mathbf 0}

\newcommand{\I}{\mathbf I}
\graphicspath{{./figs/}}

\begin{document}
\title{{\Huge Generalized Wireless-Powered Communications: When to Activate Wireless Power Transfer? } }
\author{\IEEEauthorblockN{Qingqing Wu,  Guangchi Zhang,   \IEEEauthorblockN{Derrick Wing Kwan Ng},   Wen Chen, and Robert Schober
\thanks{Q. Wu and W. Chen  are with Department of Electronic Engineering, Shanghai Jiao Tong University, email: elewuqq@nus.edu.sg, wenchen@sjtu.edu.cn.   G. Zhang is with  the School of Information Engineering, Guangdong University of Technology, email: gczhang@gdut.edu.cn. D. W. K. Ng is with the School of Electrical Engineering and Telecommunications, The University of New South Wales, email: w.k.ng@unsw.edu.au. Robert Schober is with the Institute for Digital Communications, Friedrich-Alexander-University Erlangen-N\"urnberg (FAU), email: robert.schober@fau.de. G. Zhang is supported by NSFC 61571138, STPP of Guangdong 2017B090909006, 2018A050506015, 2019B010119001, and STPP of Guangzhou 201803030028,  201904010371.  D. W. K. Ng is supported by funding from the UNSW Digital Grid Futures Institute, UNSW, Sydney, under a cross disciplinary fund scheme and  by the Australian Research Council's Discovery Project (DP190101363). W. Chen is supported by NSFC\#61671294, and by STCSM\#16JC1402900 and Grant 17510740700. } } }
\maketitle
\begin{abstract}
Wireless-powered communication network (WPCN) is a key technology to power energy-limited massive devices, such as on-board wireless sensors in autonomous vehicles, for Internet-of-Things (IoT) applications. Conventional WPCNs rely only on dedicated downlink wireless power transfer (WPT), which is practically inefficient due to the significant energy loss in wireless signal propagation. Meanwhile, ambient energy harvesting is highly appealing as devices can scavenge energy from various existing energy sources (e.g., solar energy and cellular signals). Unfortunately, the randomness of the  availability of these energy sources cannot guarantee stable communication services. Motivated by the above, we consider a generalized WPCN  where the devices can not only harvest energy from a dedicated multiple-antenna power station (PS), but can also exploit stored energy stemming from ambient energy harvesting.  Since the dedicated WPT consumes system resources,  if the stored energy is sufficient, WPT may not be needed to maximize the weighted sum rate (WSR). To analytically  characterize this phenomenon, we derive the condition for  WPT activation and reveal how it is affected by the different system parameters. Subsequently, we further derive the optimal resource allocation policy for the cases that WPT is activated and deactivated, respectively. In particular, it is found that when WPT is activated, the optimal energy beamforming at the PS does not depend on the devices' stored energy, which is shown to lead to a new unfairness issue.  Simulation results verify our theoretical findings and demonstrate the effectiveness of the proposed optimal resource allocation.
\end{abstract}

\begin{keywords}
Energy-constrained IoT networks, when to activate WPT, energy beamforming, optimal resource allocation.
\end{keywords}
\vspace{-0.3cm}
\section{Introduction}
Future vehicles are envisioned to be equipped with massive numbers of  on-board sensors to achieve reliable inter-vehicle communications and accurate navigation. {To eliminate the inconvenience caused by conventional manual battery charging and tangled wires, wireless power transfer (WPT) has gained an unprecedented upsurge of interest due to its capability to provide devices with controllable amounts of energy via radio frequency (RF) signals  \cite{ wu2016overview,niyato2017wireless,wu2018spectral}. The advancement of vehicular networks has also fueled the development of various applications based on WPT, e.g., unmanned aerial vehicles (UAV) with WPT functionality may power Internet-of-Things (IoT) devices (e.g., sensors embedded in bridges) \cite{wu2018common}.
As such,  WPT is a promising technology for achieving self-sustainable and  scalable machine-type communications \cite{sarigiannidis2017connectivity} and will have a significant impact on future IoT networks  \cite{triantafyllou2018network}.
Recently, wireless-powered communication networks (WPCNs) have been proposed in \cite{ju14_throughput}  where a dedicated power station (PS) first broadcasts energy signals to devices for downlink (DL) WPT and then the devices exploit the harvested energy for uplink wireless information transmission (UL WIT). WPCNs have been extended to different practical scenarios such as multi-antenna \cite{liang2017online,boshkovska2017robust}, multi-carrier \cite{zhang2018wireless,vamvakas2017adaptive},  secure relay \cite{zhou2019secure}, and cognitive networks \cite{kalamkar2016resource,kang2018riding}. In particular, a distributed power control algorithm is proposed in \cite{vamvakas2017adaptive} to maximize the user energy efficiency. The fundamental tradeoff between sum-rate and fairness is revealed by studying resource allocation schemes with different design objectives in \cite{kalamkar2016resource}. However, all of these works assume that the communication system relies  solely on  WPT for its energy supply.}

\begin{figure}[!t]
\centering
\includegraphics[width=3in]{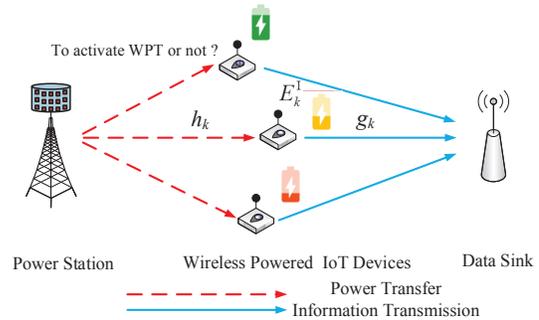}
\caption{The system model of a generalized WPCN.}\label{system:model}\vspace{-0.7cm}
\end{figure}

{Despite the evident benefits, dedicated DL WPT in WPCNs  has some disadvantages in practice. First, due to the severe signal attenuation in wireless channels, a significant amount of the energy is lost during dedicated DL WPT, which contradicts the green vision of future wireless networks \cite{wu2016overview,qing15_wpcn_twc,zhang2016fundamental}.} Second, DL WPT generally requires a dedicated  time/frequency resource.   Although the full-duplex protocol can be leveraged to support  concurrent WPT and WIT,  this may not be feasible for IoT devices due to the resulting high interference, energy consumption, and cost. In fact, the most recent narrowband IoT (NB-IoT) standard requires devices only to support half-duplex communication.   Besides relying on dedicated DL WPT that proactively generates energy signals, it is practically appealing for IoT devices to scavenge energy from ambient energy sources that are already in the surrounding environment \cite{akan2018internet}, thus incurring no additional energy consumption. For example,  devices can harvest energy from  renewable energy sources (e.g., solar) or existing RF signals (e.g.,  cellular and broadcast radio signals). However, due to the randomness and intermittence  of  ambient energy sources, devices which rely solely on ambient energy sources may not be able to harvest a sufficient amount of energy to enable  communication. {As such, designing hybrid power solutions by exploiting both dedicated WPT and ambient energy harvesting  is an interesting option for realizing sustainable and green IoT networks \cite{akan2018internet,yildiz2018hybrid,qing15_cl}. In particular, a cross-layer design strategy for incorporating different energy sources is proposed in \cite{akan2018internet}. However,  a quantitative mathematical analysis and optimization of the energy consumption and communication performance is missing. To enable efficient WPCNs,  the maximum energy consumption among all users is minimized by optimizing the transmit power and data scheduling in \cite{yildiz2018hybrid} and  the system energy efficiency is maximized by optimizing the time allocation in \cite{qing15_cl}. In both cases,  a hybrid energy harvesting model is considered. However, only numerical results are presented, which do not provide insights for system design.  In addition, energy beamforming is not exploited in \cite{akan2018internet,qing15_cl,yildiz2018hybrid}, but is very appealing for improving the efficiency of WPT.  }

{ In this paper, we propose a generalized WPCN, see Fig. 1, where each device is assumed to have a certain amount of initial energy and is  able to harvest wireless energy from a  multiple-antenna PS via DL WPT. In particular, the initial energy originates  from harvesting of ambient energy sources during the previous transmission periods and is available at the beginning of the current transmission period. To ensure resource allocation fairness among all devices, we aim to maximize the weighted sum rate (WSR) of the devices in the considered  transmission period, taking into account the device circuit power consumption which is critical for short-range applications such as IoT \cite{akan2018internet,wu2016overview}.   For this generalized system setup, we answer the following two fundamental questions: 1) When should DL WPT  be activated in  generalized WPCNs and how is this affected by the system parameters? This question naturally arises from the fact that  DL WPT may not be needed due to  the initial energy and dedicating a fixed large period to DL WPT can  be harmful to the overall performance. To the best of the authors'  knowledge, this is the first work that attempts to answer this question. 2) What is the optimal resource allocation for generalized WPCNs? Existing solutions \cite{ju14_throughput}, which consider neither the initial energy nor the device circuit power consumption, are  suboptimal for the considered problem. Based on the answers to these two questions, we show that as long as WPT is activated, both the DL  energy beamforming at the PS and the UL transmit power of the devices are independent of  the initial energy. However, adopting the WSR as objective function provides flexibility to adjust the devices' harvested energy in accordance with their initial energy.
 }

{The remainder of the paper is organized as follows. In Section II, the system model and problem formulation are presented. In Section III, we derive the  condition for WPT activation as well as the optimal resource allocation solution. Numerical results are provided in Section IV, and the paper is concluded in Section V.}
\vspace{-0.3cm}
\section{System Model and Problem Formulation}
\subsection{System Model}
As shown in Fig. \ref{system:model},  we consider a generalized WPCN where a PS  equipped with $M$ antennas is deployed to assist $K$ single-antenna IoT devices  in communicating with a single-antenna data sink (DS). {Specifically, the $K$ devices are assumed to collect energy from two energy sources, namely, an ambient energy source (e.g., a renewable energy sources or TV signals) and  dedicated DL WPT. Due to the randomness of the ambient energy source, we assume that for each device,  a certain amount of energy $E^{\rm I}_k\geq 0$ Joule (J) is available at the beginning of each transmission period. The considered model can also be applied to the case when each device in the WPCN has a conventional battery charged with a limited amount of energy, i.e., $E^{\rm I}_k$ is drawn from the battery in each transmission period, and WPT is used to prolong the battery lifetime.}
Furthermore, when DL WPT is activated to replenish the energy of the devices,  the ``harvest and then transmit'' scheme is adopted \cite{ju14_throughput}, i.e.,  the  devices first harvest energy from the signal sent from the PS in the DL and then transmit information to the DS in the UL.    For low-complexity implementation as well as low circuit power consumption, it is assumed that all devices adopt  the time division multiple access (TDMA) protocol \cite{ju14_throughput}. In addition, we assume a quasi-static flat-fading channel model for all involved channels where  the channel state information (CSI) is perfectly known at the PS. { In particular,  the DL channel from the PS to device $k$ and the UL channel from device $k$ to the DS are denoted by $\h^H_k\in \mathbb{C}^{1\times M}$ and $g_k$, respectively.}

During the DL WPT phase, the PS broadcasts an energy signal ${\x}\in \mathbb{C}^{M\times 1}$ for a duration $\tau_0$ where  the {transmit covariance matrix ${\W}= \mathbb{E}({\x}{\x}^H)$} is subject to the maximum transmit power constraint ${\rm{Tr}}({\W})\leq P_{\max}$. Here,   $\mathbb{E}(\cdot)$ and $\rm{Tr}(\cdot)$ denote  statistical  expectation and the trace  of a matrix, respectively. By ignoring the negligible receiver noise power for energy harvesting, {the amount of energy harvested at device $k$ under a linear energy harvesting model\footnote{{Based on Fig. 3 in \cite{boshkovska2017robust}, one can  verify that the linear energy harvesting model is appropriate for the system parameters considered in Section IV. In practice, adopting a non-linear energy harvesting model  is more general and constitutes an interesting topic for future research.} } \cite{ju14_throughput,liang2017online,zhang2018wireless,kalamkar2016resource,mishra20172} can be expressed as
\begin{align}\label{eq3}
E^h_k=\eta_k\tau_0 \mathbb{E}(\h_k^H{\x}{\x}^H\h_k) =\eta_k\tau_0{\rm{Tr}}(\W \HH_k),
\end{align}}where $\eta_k \in (0,1]$ is the energy harvesting efficiency \cite{ju14_throughput} of device $k$ and {$\HH_k=\h_k\h_k^H$}. During the UL WIT phase, each energy harvesting device transmits an independent information signal to the DS for a duration of $\tau_k$ and with transmit power $p_k$. Accordingly, the achievable rate of device $k$ in bits/Hz  can be expressed as
\begin{align}\label{eq6}
r_k=\tau_k \log_2\left(1+p_k\gamma_k\right),
\end{align}
where $\gamma_k= \frac{|g_k|^2}{\sigma^2}$ is the normalized UL channel gain of device $k$ and  $\sigma^2$ is the additive white Gaussian noise power at the DS. Besides the transmit power, a constant circuit power accounting for the operation of the electronic circuits of the transmitter  (e.g., mixer and frequency synthesizers) is also consumed at device $k$, and is denoted by $p_{{\rm{c}},k}$ \cite{wu2016overview,qing15_wpcn_twc,zhang2016fundamental}.
\begin{remark}
{The considered model can be applied to  UAV-enabled WPCNs where a UAV PS is deployed at a suitable location to power energy-limited IoT devices within a certain area \cite{park2018minimum,zeng2019accessing,suman2018uav,suman2018path,wu2018joint}.}
\end{remark}

\vspace{-0.3cm}
\subsection{Problem Formulation}
{Our objective is to maximize the WSR of all devices, i.e., $R_{s}=\sum_{k=1}^Kw_kr_k$,  where  $w_k$ denotes the weight of device $k$. By varying the values of the weights, the system designer is able to set different priorities and enforce certain notions of fairness among devices. Specifically, we jointly optimize the DL and UL time allocation, the transmit covariance matrix at the PS, and the power control at the devices, which leads to the following optimization problem:}
\begin{align}\label{eq10}
 \mathop {\max }\limits_{{\tau_{0}, \{\tau_{k}\},\{p_{k}\}, \W } }&~~ \sum_{k=1}^{K}w_k\tau_k\log_2\left(1+p_k\gamma_k\right) \nonumber \\
\text{s.t.} ~~~~~&  \text{C1:}~ \left({p_k}+p_{{\rm{c}},k}\right)\tau_{k}\leq \eta_k\tau_0{\rm{Tr}}(\W \HH_k) +E^{\rm I}_k, ~ \forall k, \nonumber \\
&\text{C2:}~{\rm{Tr}}(\W)\leq P_{\mathop{\max}}, \W\succeq \0, \nonumber \\
&\text{C3:}~\tau_{0}+\sum_{k=1}^{K}\tau_k\leq T, \nonumber \\
& \text{C4:}~\tau_{0}\geq0, ~  \tau_k\geq  0, p_k\geq  0, ~\forall k.
\end{align}
In \eqref{eq10}, C1 ensures that the total energy consumed at  each device cannot exceed the total available energy and C2 reflects the maximum transmit power constraint at the PS.  C3 and C4 are the total transmission time constraint and the non-negative constraints on the optimization variables, respectively. Note that when $\tau_0=0$, this means that WPT is not activated and the devices only perform UL WIT.
 However, problem \eqref{eq10} is non-convex due to the coupling of the optimization variables in C1. To circumvent this difficulty,   we apply a changes of variable as  $e_k=\tau_kp_k$ and $\V= \tau_0\W$ and rewrite problem (\ref{eq10})  as
\begin{align}\label{eq11}
 \mathop {\max }\limits_{{\tau_{0}, \{\tau_{k}\},\{e_{k}\}, \V } }&~~\sum_{k=1}^{K} w_k \tau_k\log_2\left(1+\frac{e_k}{\tau_k}\gamma_k\right) \nonumber \\
\text{s.t.} ~~~~~&  \text{C1:}~ {e_k}+p_{{\rm{c}},k}\tau_{k}\leq \eta_k{\rm{Tr}}(\V \HH_k) +E^{\rm I}_k, ~ \forall k, \nonumber \\
& \text{C2:}~ {\rm{Tr}}(\V)\leq \tau_0P_{\mathop{\max}}, \V\succeq \0, \nonumber \\
&\text{C3}, ~ \text{C4:}~\tau_{0}\geq0, ~  \tau_k\geq  0, ~e_k\geq  0, ~\forall k.
\end{align}
Although (\ref{eq11}) is a convex optimization problem that can be solved by standard solvers, e.g., the interior point method,  these numerical approaches cannot provide any useful insights into the optimal solution. In the next section, we exploit the inherent structure of problem \eqref{eq11} to reveal how the system can maximize the WSR, which also results in an optimal and efficient solution.

\section{Optimal Solution}
 In this section, we first study some properties of the optimal  energy and time utilization as well as the user scheduling in the generalized WPCNs. Then, we derive the DL WPT activation condition as well as the optimal solution to problem \eqref{eq11}.

{First,  for generalized WPCNs,  all devices satisfying $ \eta_k{\rm{Tr}}(\V \HH_k) +E^{\rm I}_k>0$, $\forall k$, will be scheduled  to completely deplete their energy in UL WIT, i.e., $\tau_k>0$ and $ \left({p_k}+p_{{\rm{c}},k}\right)\tau_k=  \eta_k{\rm{Tr}}(\V \HH_k) +E^{\rm I}_k$. This is because the energy in different devices cannot be shared and thus scheduling more devices and depleting their energy always helps increase the WSR. Second,  the total available transmission time is always exhausted for DL WPT and UL WIT, i.e., $\tau_{0}+\sum_{k=1}^{K}\tau_k= T$, because  if  $\tau_{0}+\sum_{k=1}^{K}\tau_k<T$,  the WSR can be further improved by properly increasing $\tau_{0}$ and $\tau_k$, $\forall k$. The above two properties will be used to simplify problem \eqref{eq11} in the following. }

\subsection{When to Activate DL WPT?}
Intuitively, if the available initial energy is sufficient at the IoT devices, the DL WPT does not need to be activated. Now, we reveal the optimal  condition for activating WPT.
\begin{proposition}\label{theom_activation}
For generalized WPCNs, DL WPT will be activated \emph{if and only if} the following condition is satisfied,
{\begin{align}\label{eq_activate401}
T>  \sum_{k=1}^{K}\frac{E^{\rm I}_k}{p^*_k +p_{c,k}},
\end{align}}
 where $p^*_k$ is the unique root of
{ \begin{align}\label{eq_activate402}
\mathcal{G}^{\rm on}_k(p_k) & \triangleq w_k\log_2\left(1 + p_k\gamma_k\right)-\frac{w_k(p_k+p_{{\rm c},k})\gamma_k}{(1+p_k\gamma_k)\ln2}  \nonumber\\
&\quad\, -  P_{\mathop{\max}}\psi\left(\sum_{k=1}^{K} \frac{w_k\eta_k\gamma_k}{(1 + p_k\gamma_k)\ln2}{\HH}_k\right) = 0
 \end{align}}
and $\psi(\X)$ denotes the maximum eigenvalue of matrix $\X$.
\end{proposition}
\begin{proof}
 Please refer to  Appendix A.
\end{proof}

{Proposition \ref{theom_activation} explicitly answers the question when DL WPT is needed to maximize the WSR. The result in \eqref{eq_activate401} confirms  the intuition that less initial energy at the devices, i.e., smaller $E^{\rm I}_k$,  will make WPT more likely to be activated and vice versa. For example, in the extreme case when all devices do not have any initial energy, i.e., $E^{\rm I}_k=0$, $\forall k$, \eqref{eq_activate401} always holds and thus DL WPT is always activated, which corresponds to conventional WPCNs that only rely on DL WPT \cite{ju14_throughput, liang2017online,zhang2018wireless}. Note that the $K$ transcendental equations in   \eqref{eq_activate402} need to be solved jointly in order to obtain the $K$ variables, $p^*_k$'s, and are difficult to solve in general. The key observation is that $\psi(\sum_{k=1}^{K} \frac{w_k\eta_k\gamma_k}{(1 + p_k\gamma_k)\ln2}{\HH}_k)$ is a common term for all $K$ devices and $\mathcal{G}^{\rm on}_k(p_k)$  increases monotonically with $p_k$, $\forall k$. As such, $p^*_k$ can be obtained readily by applying the bisection method.
 \begin{remark}\label{remark1}
The DL WPT activation condition in (\ref{eq_activate401}) has an interesting interpretation. In particular, given the hardware parameters (i.e.,  $P_{\max}$,  $\eta_k$,
$p_{{\rm c},k}$) as well as  the channel parameters (i.e., $\HH_k$ and $\gamma_k$), the system can calculate a ``virtual'' transmit power $p^*_k$ for each device $k$ according to (\ref{eq_activate402}). By letting the devices transmit at this virtual power,  the UL WIT time needed for depleting all initial energy can be obtained as $\widehat{T}\triangleq\sum_{k=1}^{K}\frac{E^{\rm I}_k}{p^*_k+p_{{\rm c}, k}}$. Thus,  if $T> \widehat{T}$, this  means that the initial energy is not sufficient and  DL WPT should be turned on in the rest of the transmission time period so as to transfer more energy to the devices for achieving a higher WSR. In contrast, if $T\leq  \widehat{T}$, this means that the initial energy is already sufficient for UL WIT and hence DL WPT should be turned off to save transmission time for UL WIT. Therefore, whether DL WPT is activated or not,  fundamentally depends on whether the available transmission time period or the available initial energy are the bottleneck for the WSR of the considered system.
\end{remark}

It is worth pointing out that the explicit expressions in \eqref{eq_activate401} and \eqref{eq_activate402} also facilitate the characterization of the impact of the hardware parameters on the DL WPT activation.
Let $\mathcal{F}( P_{\max}, \eta_k, p_{{\rm c}, k}) \triangleq T-  \sum_{k=1}^{K}\frac{E^{\rm I}_k}{p^*_k +p_{{\rm c},k}}$.
\begin{proposition}\label{monotonic1}
$\mathcal{F}( P_{\max}, \eta_k, p_{{\rm c}, k})$ is an increasing function of  $P_{\max}$, $\eta_k$,  and $p_{{\rm c}, k}$, respectively.
\end{proposition}
\begin{proof}
Due to page limitation, we only provide the proof for the case of $P_{\max}$, while the other cases can be similarly shown. Since  $\mathcal{F}$ does not explicitly depend on $P_{\max}$, we first need to identify the relationship between $P_{\max}$ and $p^*_k$ based on \eqref{eq_activate402}. Since $\sum_{k=1}^{K} \frac{w_k\eta_k\gamma_k}{(1 + p_k\gamma_k)\ln2}{\HH}_k\succeq \0$, we have
{\begin{align}
 \frac{w_k\eta_k\gamma_k}{(1 + p'_k\gamma_k)\ln2}{\HH}_k \preceq  \frac{w_k\eta_k\gamma_k}{(1 + p_k\gamma_k)\ln2}{\HH}_k
\end{align}}
for $p'_k\geq p_k$. Note that $\psi(\X)\geq\psi(\X') $ if $\X \succeq \X' \succeq\0$. As such,  the last term of $\mathcal{G}^{\rm on}_k(p_k)$ is monotonically non-decreasing with respect to $p_k$. In addition, it can be easily verified that the first two terms of $\mathcal{G}^{\rm on}_k(p_k)$ monotonically increase with $p_k$.  As such, it can be concluded that $\mathcal{G}^{\rm on}_k(p_k)$ is an increasing function of $p_k$. On the other hand, it is easy to observe that $\mathcal{G}^{\rm on}_k(p_k)$ decreases with  $P_{\max}$. Thus, the solution of $\mathcal{G}^{\rm on}_k(p_k)=0$, i.e., $p^*_k$, monotonically increases with $P_{\max}$, which thus also holds for $\mathcal{F}( P_{\max}, \eta_k, p_{{\rm c}, k})$.
\end{proof}

The insight behind Proposition \ref{monotonic1} is that as $P_{\max}$, $\eta_k$, and/or $p_{{\rm c}, k}$ become larger, it becomes more likely that DL WPT is activated.  This also agrees with intuition. On the one hand, higher $P_{\max}$ and  $\eta_k$ allow the devices to harvest more energy from DL WPT, which in turn effectively reduces the time spent on dedicated DL WPT for harvesting a certain required amount of energy. In the extreme case, when $P_{\max}$ is sufficiently large, then the time needed for DL WPT is negligibly small and activating DL WPT is always desirable for maximizing the WSR.
On the other hand, a larger $p_{{\rm c}, k}$ implies more energy consumption in UL WIT, which renders the initial energy more likely to be insufficient and activating DL WPT becomes necessary.

\subsection{Optimal Resource Allocation}
Now, we study the optimal resource allocation as follows:
\begin{theorem}\label{theorem:2}
If DL WPT is activated, the optimal transmit covariance matrix and  time allocation are given by
{ \begin{align}
   \W^*& =P_{\max}\ve\ve^{H}, ~~ \tau^*_0 =\frac{T - {\sum_{k=1}^{K}\frac{E^{\rm I}_k\gamma_k}{p^*_k+p_{{\rm c},k}}}}{1 +\sum_{k=1}^{K}\frac{\eta_k{ {\rm{Tr}} (\W^*{\HH}_k)}\gamma_k }{p^*_k+p_{{\rm c},k}}},  \label{tau0}\\
  \tau^*_k& = \frac{\eta_k{\rm{Tr}}(\W^* \HH_k)\tau^*_0 +E^{\rm I}_k}{p^*_k+p_{{\rm c},k}}, \forall k, \label{tauk}
 \end{align}}where $p^{\star}_k$ is the  root in \eqref{eq_activate402} and $\ve$ is the principle eigenvector of matrix  $\sum_{k=1}^{K} w_k \frac{\eta_k\gamma_k}{ (1 + p^*_k\gamma_k)\ln2}\mathbf{H}_k$. In contrast, if DL WPT is not activated, the optimal time allocation  is given by
  {\begin{align}
 \tau_0^{\star} =0, ~~  \tau^{\star}_k = \frac{E^{\rm I}_k}{p^{\star}_k+ p_{{\rm{c}},k} }, \forall k, \label{tauk10}
 \end{align} }
  where $p^{\star}_k$, $\forall k$, is the unique root of
{ \begin{align}\label{eq_activate23}
\mathcal{G}^{\rm off}_k(p_k) \triangleq w_k\ln(1+p_{k}\gamma_k)-\frac{w_k(p_k+p_{{\rm c},k})\gamma_k}{(1+p_k\gamma_k)\ln2}=\delta^{\star}
 \end{align} }
and $\sum_{k=1}^{K}\frac{E^{\rm I}_k}{(\mathcal{G}^{\rm off}_k)^{-1}(\delta^{\star})+ p_{{\rm{c}},k}  }=T$.
\end{theorem}
\begin{proof}
 Please refer to  Appendix B.
\end{proof}

From  \eqref{tau0}, it is observed that when DL WPT is activated, the optimal  $\W$ is rank one and independent of the users' initial energy $E^{I}_k$. This suggests that to maximize WSR, the optimal energy beamforming direction does not depend on the initial energy levels of the devices,  $E_k^{\rm I}$. In fact, this also leads to a new kind of unfairness among the devices as devices with less initial energy may expect to harvest more energy during DL WPT.  {However, this problem can be resolved by adjusting the weights of the different devices in the objective function in a suitable manner. To illustrate this, we consider an extreme case. By setting the weight of device $k$ sufficiently larger than those of the other devices, i.e., $w_k \gg w_m$, $\forall m\neq k$, we have
{\begin{align}
\sum_{m=1}^{K} \frac{w_m\eta_m\gamma_m}{1 + p^*_m\gamma_m}\mathbf{H}_m\approx \frac{w_k\eta_k\gamma_k}{1 + p^*_k\gamma_k}\mathbf{H}_k
\end{align}}and $\W \approx  {P_{\max}}\h_k\h_k^{H}/{\|\h_k\|^2}$, which suggests that the PS steers its energy beam towards the exact direction of device $k$ so as to maximize its harvested energy, i.e., maximum-ratio transmission (MRT) is performed.  This suggests that the achievable rates of different devices can be balanced by imposing different weights based on the respective initial energy.}  Next, we study the behavior of the device's transmit power in such energy-constrained networks.
\begin{proposition}\label{device:transmitpower}
 When DL WPT is not activated, the transmit power of each device in UL WIT decreases and increases with  $T$ and ${E^{\rm I}_k}$, respectively. When DL WPT is activated, the transmit power of each device in UL WIT is independent of  $T$ and $E^{\rm I}_k$, $\forall k$.
\end{proposition}
\begin{proof}
The result can be readily proved by analyzing the transmit powers in  \eqref{eq_activate23} and  \eqref{eq_activate402}, respectively, which is omitted for brevity.
\end{proof}

 Proposition \ref{device:transmitpower} suggests that when WPT is activated, there exists an optimal transmit power for each device. This is because for a smaller transmit power, more time is needed to deplete the total energy, which incurs more circuit energy consumption. For a larger transmit power, more energy needs to be harvested, which requires a longer time for DL WPT and in turn limits the time available for UL WIT and hence the WSR.
\begin{figure}[!t]
\centering
\includegraphics[width=2.8in]{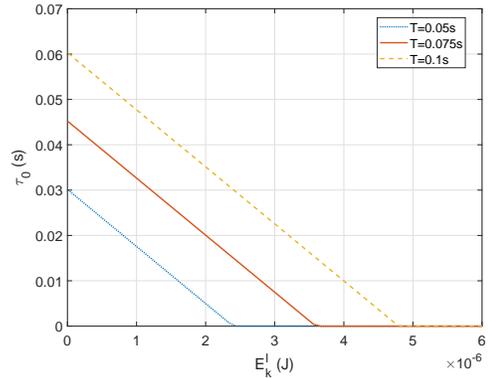}\vspace{-0.3cm}
\caption{DL WPT activation condition.}\label{fig:side:a}\vspace{-0.3cm}
\end{figure}

\begin{figure}[!t]
\centering
\includegraphics[width=2.8in]{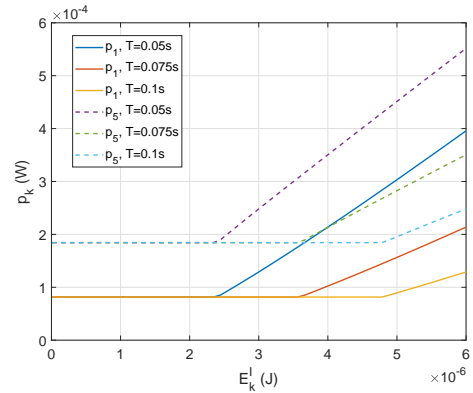}\vspace{-0.3cm}
\caption{Device transmit power versus initial energy.}\label{fig:side:c}\vspace{-0.3cm}
\end{figure}

\vspace{-0.3cm}
\section{Numerical Results }
In our simulations, the DL and UL channel gains are modeled as $10^{-3} \rho^2 d^{-\alpha}$, where $\rho^2$ is an exponentially distributed random variable with unit mean, $d$ is the link distance, and $\alpha$ is the path loss exponent. The path loss exponents of the DL and UL channels are set as  $2.2$ and $3$, respectively. The other parameters are set as follows: {$M=4$, $\eta_k=0.5$, $\forall k$\cite{mishra20172}}, $w_k=1$, $\forall k$,  $P_{\max}=40$ dBm, $E^{\rm I}_k = E^{\rm I}_m$, $\forall m \neq k$, and  $\sigma^2 = -117$ dBm,  unless specified otherwise.

\subsection{Illustration of DL WPT Activation Condition}
In Fig. \ref{fig:side:a}, we consider a concrete numerical example for 5 devices to illustrate the WPT activation in Proposition \ref{theom_activation}. We set $\rho^2=1$ for all the devices to show the optimal $\tau_0$ for different $E^{\rm I}_k$ and $T$. The DL and UL link distances are set as $10$ m and $90$ m, respectively.  {For the purpose of illustration, we  differentiate the 5 devices by adopting different device circuit powers. All other parameters are the same for all devices. The impact of the other parameters (e.g., $w_k$ and $\eta_k$) on the system performance can be studied in a similar manner.} Specifically, the circuit power of device $k$ is set as $p_{{\rm c},k} = 0.05k$ mW.   First, it is observed that for a given transmission period (e.g., $T=0.05$ s), the optimal $\tau_0$ is non-increasing as $E^{\rm I}_k$ increases. In particular, when $E^{\rm I}_k\geq 2.4\times 10^{-6}$ J, one can observe that $\tau_0= 0$ which means that WPT is no longer needed for  maximization of the WSR. In addition, for any given $E^{\rm I}_k$,  it is observed that  as $T$ becomes larger, the optimal $\tau_0$  increases accordingly. This is intuitive since a larger $T$ generally consumes more  circuit energy, which makes  WPT more likely to be activated. Correspondingly, the transmit powers of devices 1 and 5 are plotted in Fig. \ref{fig:side:c}. It is observed that when WPT is activated, the devices transmit with constant powers for UL WIT, even for different $E^{\rm I}_k$ and $T$. {In contrast, for a given $T$, the transmit power of a device is non-decreasing as  $E^{\rm I}_k$ increases, which verifies the result in Proposition \ref{device:transmitpower}.}

\subsection{Performance Comparison}
In Fig. \ref{fig:side:K}, we plot the average WSR versus the number of devices $K$ by setting $T=0.1$ s. For illustration, the WSR is normalized by $T$. The devices use the same weights $w_k=1$, $\forall k$, and are  randomly located at distances of 5-10 m away from the PS and $p_{{\rm c},k} = 0.1$ mW, $\forall k$. All other parameters are the same as in the previous example. {For comparison, we consider two benchmark schemes: 1) isotropic beamforming where  $\mathbf{W}=({P_{\max}/M})\mathbf{I}$ ($\mathbf{I}$ denotes the identity matrix) with the time allocation optimized to maximize the WSR; 2) fixed DL WPT where $\tau_0=0.5 T$ with   $\mathbf{W}$ optimized to maximize the WSR. First, it is observed that  for small $K$, the optimal solution significantly outperforms  isotropic beamforming, whereas the performance gap decreases as $K$ becomes larger. This is because a larger $K$ implies that there are more devices with available  energy. Accordingly, the DL WPT time is reduced to ensure a longer time for UL WIT, which shrinks the gain achievable by energy beamforming. On the other hand, by decreasing the initial energy at the battery,  the performance gap becomes more pronounced  for a given number of devices, which demonstrates the importance of energy beamforming for energy-limited networks.} Furthermore, one can observe that  the fixed time allocation scheme suffers from a substantial performance loss as compared to the proposed solution, which justifies the necessity of the time allocation optimization. {Finally, we show in Fig. \ref{fig:side:M} the WSR versus the number of antennas ($M$) for $K=5$, where we use the same parameters as for Fig. \ref{fig:side:K}. It is observed that the proposed solution significantly outperforms the benchmark schemes. In particular, since the isotropic beamforming scheme forces the PS  to radiate energy signal equally in all directions, the energy beamforming gain is not exploited and the WSR does not improve with the number of antennas.}

\begin{figure}[!t]
\centering
\includegraphics[width=2.8in]{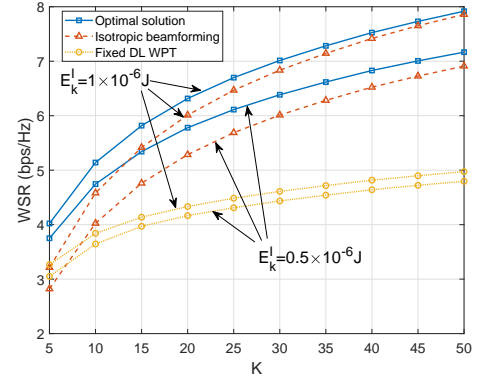}\vspace{-0.3cm}
\caption{WSR versus the number of devices. }\label{fig:side:K}\vspace{-0.3cm}
\end{figure}

\begin{figure}[!t]
\centering
\includegraphics[width=2.8in]{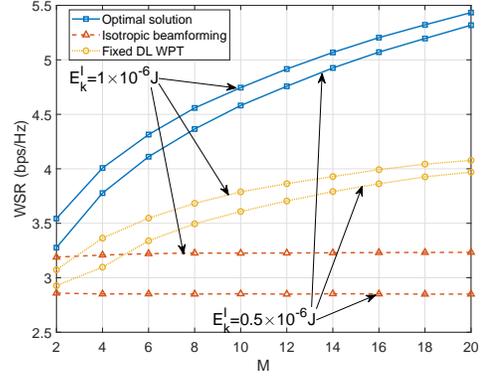}\vspace{-0.3cm}
\caption{WSR versus the number of antennas at the PS. }\label{fig:side:M}\vspace{-0.3cm}
\end{figure}

\section{Conclusions}
{In this paper, we have studied the WSR maximization problem for a generalized WPCN where the devices can exploit both  the energy harvested from DL WPT as well as their stored initial energy  to support UL WIT. In particular, we have answered two fundamental questions: when will WPT be activated and what is the optimal resource allocation? The obtained results shed light on how the optimal  energy beamformer and  transmit power maximizing the WSR depend on the system parameters. }

\appendices
\section*{Appendix A: Proof of Proposition 1}
Since  problem  (\ref{eq11}) is a convex optimization problem and also satisfies Slater's constraint qualification \cite{Boyd}, the duality gap between (\ref{eq11}) and its dual problem is zero. This implies that its optimal solution  can be obtained by analyzing the Karush-Kuhn-Tucker (KKT) conditions. Specifically, the Lagrangian function of \eqref{eq11} can be written as
\begin{align}
\mathcal{L} =&  \sum_{k=1}^{K} w_k\tau_k\log_2(1+\frac{e_k\gamma_k}{\tau_k})+\mu (\tau_0P_{\max}-{\rm{Tr}}(\V) )  \nonumber\\
&+  \sum_{k=1}^{K}\lambda_kE^{\rm I}_k   + \delta( T -  \tau_0 -\sum_{k=1}^{K}\tau_k)  \nonumber\\
& + \sum_{k=1}^{K}\lambda_k(\eta_k{\rm{Tr}}(\V \HH_k) - {e_k}-p_{{\rm{c}},k}\tau_{k}),
\end{align}
where $\lambda_k$, $\mu$,  and $\delta$ are the Lagrange multipliers associated with C1, C2, and C3, respectively.   By taking the partial derivative of $\mathcal{L}$ with respect to $\tau_0$, $\tau_k$, and  $e_k$, respectively, we  obtain
\begin{small}\begin{align}
\frac{\partial\mathcal{L}}{\partial \tau_0}&= \mu P_{\max} - \delta, \label{apdx_eq20}\\
\frac{\partial\mathcal{L}}{\partial \tau_k}&=w_k\log_2(1+\frac{e_k}{\tau_k}\gamma_k)-\frac{w_ke_k\gamma_k}{(\tau_k+e_k\gamma_k)\ln2}-\lambda_kp_{{\rm c},k}-\delta, \label{apdx_eq21}\\
\frac{\partial\mathcal{L}}{\partial e_k}&=\frac{w_k\tau_k\gamma_k}{(\tau_k+e_k\gamma_k)\ln2}-\lambda_k. \label{apdx_eq22}
\end{align}\end{small}{If DL WPT is activated for the optimal solution, i.e.,$\tau_0>0$, we have $\tau_k>0$ and $e_k>0, \forall k$. This implies  that $\frac{\partial\mathcal{L}}{\partial \tau_0}=0$, $\frac{\partial\mathcal{L}}{\partial \tau_k}=0$, and $\frac{\partial\mathcal{L}}{\partial e_k}=0$, respectively, which yields}
 {\begin{align}\label{apdx:eq_activate402}
 w_k\log_2\left(1 + p^*_k\gamma_k\right)-\frac{w_k(p^*_k+p_{{\rm c},k})\gamma_k}{ (1+p^*_k\gamma_k)\ln2} -\mu P_{\max}= 0.
 \end{align} }Furthermore, $\mathcal{L}$ can be rewritten as $\mathcal{L} ={\rm{Tr}}(\A\V) +\Delta \mathcal{L}$ where $\A = \B-\mu\I$,  $\B= \sum_{k=1}^{K}\lambda_k\eta_k\HH_k$, and $\Delta \mathcal{L}$ includes all terms in $\mathcal{L}$ that are independent of $\V$. Denote the maximum eigenvalue of  $\B$ as $\psi(\B)$.  To ensure that $\mathcal{L}$ is bounded from  above, we obtain  $\A\preceq \0$, which implies that $\psi(\B)\leq \mu$. However, for the case $\psi(\B)< \mu$, we have $\A\prec  \0$ and hence $\V=\0$, which contradicts that the WPT should be activated at the optimal solution. Thus, it follows that  $\psi(\B)= \mu$ and combining this with \eqref{apdx_eq22} and \eqref{apdx:eq_activate402} yields \eqref{eq_activate402}.  Since each device depletes all of its available energy, we have $\tau_k = \frac{  \eta_k\tau_0{\rm{Tr}}(\W \HH_k)  +E^{\rm I}_k}{p^*_k +p_{{\rm c},k} }$. Substituting $\tau_k$ into the total time constraint, i.e.,  $\tau_0 +\sum_{k=1}^{K}\tau_k = T$, yields
 \begin{align}\label{apdx:timeconstraint}
\tau_0\left(1+ \sum_{k=1}^{K}\frac{  \eta_k{\rm{Tr}}(\W \HH_k) }{p^*_k +p_{{\rm c},k} } \right)  +  \sum_{k=1}^{K}\frac{E^{\rm I}_kg_k}{p^*_k +p_{{\rm c},k}}  =T.
 \end{align}
As such,  the condition $T>  \sum_{k=1}^{K}\frac{E^{\rm I}_k}{p^*_k +p_{{\rm c},k}}$ in \eqref{eq_activate401} must hold  if DL WPT is activated, i.e., $\tau_0>0$. On the other hand, if  the condition in  \eqref{eq_activate401}  holds, a non-trivial $\tau_0>0$ can be obtained from the above, which satisfies the KKT conditions and is thus optimal for problem \eqref{eq11} due to its convexity.
\vspace{-0.2cm}

\section*{Appendix B: Proof of Theorem \ref{theorem:2}}\label{apdix_theorem2}
First, if  WPT  is activated, it follows from Appendix A that the optimal $\V$ maximizing the Lagrangian function $\mathcal{L}$ should satisfy the following KKT conditions:
\begin{align}
{\rm tr}(\A\V) = 0, \label{tracezero}\\
\mu (\tau_0P_{\max}-{\rm{Tr}}(\V) ) =0. \label{complementary}
\end{align}
Since $\A \preceq \0$ and $\V\succeq \0$, \eqref{tracezero} implies $\A\V =\0$ and hence $\V$ lies in the null space of $\A$. In addition, since  $\psi(\B) =\mu$, it follows that ${\rm{rank}}(\A)=M-1$ and the null space of $\A$ is spanned by the largest eigenvector of $\B$, denoted by
$\ve$, i.e., $\A\ve=\0$. As such, $\V$ can be written as $\beta \ve\ve^H$ where $\beta$ is a scaling factor for satisfying the maximum power constraint. Since $\mu>0$, \eqref{complementary} implies $\tau_0P_{\max}-{\rm{Tr}}(\V)=\tau_0P_{\max}-\beta{\rm{Tr}}(\ve\ve^H)=0$. Thus, we have $\beta=\tau_0P_{\max}$ and $\V= \tau_0P_{\max}\ve\ve^H$. Note that from $\V= \tau_0\W$, we can obtain $\W$ as in \eqref{tau0}. Substituting $\V$ into \eqref{apdx:timeconstraint} and combining this with the fact that C1 is satisfied with equality, the time allocation $\tau_0$ and $\tau_k$ can be readily obtained as in \eqref{tau0} and \eqref{tauk}, respectively. Second, if WPT is not activated, i.e., $\tau_0=0$,  we have $\frac{\partial\mathcal{L}}{\partial \tau_0}\leq 0$, $\frac{\partial\mathcal{L}}{\partial \tau_k}=0$, and $\frac{\partial\mathcal{L}}{\partial e_k}=0$, respectively, yielding
 {\begin{align}\label{apdx:eq:notactivate}
{ w_k\ln(1+p^{\star}_{k}\gamma_k)-\frac{w_k(p^{\star}_k+p_{{\rm c},k})\gamma_k}{(1+p^{\star}_k\gamma_k)\ln2}  = \delta^{\star},}
 \end{align}}where $\delta^{\star}$ is the optimal dual variable. Note that the time constraint also holds with equality at the optimal solution, i.e., $\sum_{k=1}^{K}\tau_k=\sum_{k=1}^{K}\frac{E^{\rm I}_k}{p^{\star}_k+ p_{{\rm{c}},k}  }=T$, which results in \eqref{tauk10} and \eqref{eq_activate23}.
\bibliographystyle{IEEEtran}
\bibliography{IEEEabrv,mybib}

\end{document}